\documentclass[aps, pra,reprint, showpacs,nofootinbib,
superscriptaddress,10pt,floatfix,longbibliography]{revtex4-2}
\usepackage[utf8]{inputenc}
\usepackage{amsmath, amssymb, braket, mathtools, amsthm}
\usepackage[qm]{qcircuit}
\usepackage{graphicx}
\usepackage{ragged2e}
\usepackage[ruled,vlined]{algorithm2e}
\usepackage{bm}
\usepackage[normalem]{ulem}
\usepackage{physics}
\usepackage{hyperref}
\usepackage{orcidlink}
\usepackage[caption=false]{subfig} 
\usepackage{xcolor}     
\usepackage{microtype}         
\usepackage{natbib}
\usepackage{ulem}
\usepackage{soul}
\usepackage{etoolbox}    
\AtBeginEnvironment{thebibliography}{%
  \sloppy
  \setlength\emergencystretch{1em}%
}
\AtEndEnvironment{thebibliography}{\fussy}

\newcommand{\stretchedhat}[1]{%
\savestack{\tmpbox}{\stretchto{%
  \scaleto{%
    \scalerel*[\widthof{\ensuremath{#1}}]{\kern.1pt\mathchar"0362\kern.1pt}%
    {\rule{0ex}{\textheight}}
  }{\textheight}%
}{2.4ex}}%
\stackon[-6.9pt]{#1}{\tmpbox}%
}
\newcommand{\stretchedtilde}[1]{%
\savestack{\tmpbox}{\stretchto{%
  \scaleto{%
    \scalerel*[\widthof{\ensuremath{#1}}]{\kern.1pt\mathchar"307E\kern.1pt}%
    {\rule{0ex}{\textheight}}
  }{\textheight}%
}{2.4ex}}%
\stackon[-6.9pt]{#1}{\tmpbox}%
}
\newtheorem{theorem}{Theorem}
\newtheorem{proposition}{Proposition}

\newtheorem{lemma}{Lemma}
\newtheorem{corollary}{Corollary}

\newtheorem{property}{Property}
\newcommand{\imi}{\mathrm{i}}

\newcommand{\PTr}[2]{\text{Tr}_{#1}\left(#2\right)}

\newcommand{\sign}[1]{\mathsf{sign}(#1)}
\newcommand{\Hil}{{\mathcal H}}
\newcommand{\Th}{\text{th}}
\newcommand{\J}{{J\vphantom{\overline{J}}}}

\newcommand{\Jbar}{{\overline{J}}}

\newcommand{\slpo}{S_{\scalebox{0.5}{CHSH}}^{\scalebox{0.5}{LPO}}}
\newcommand{\clpo}[1]{C^{\scalebox{0.5}{LPO}}_{#1}}
\newcommand{\schsh}{S_{\scalebox{0.5}{CHSH}}}
\newcommand{\wit}[2][]{%
  \mathrm{W}%
  \if\relax\detokenize{#1}\relax\else^{#1}\fi%
  _{\scalebox{0.5}{\text{#2}}}(\rho)%
}
\makeatother
\definecolor{rkrPurple}{HTML}{73024F}

\begin{document}

\title{Local perception operators and classicality: new tools for old tests}
\author{Rohit Kishan Ray\orcidlink{0000-0002-5443-4782}}
\email{rkray@vt.edu}
\affiliation{Center for theoretical Physics of Complex Systems, Institute for Basic Science (PCS-IBS), Daejeon - 34126, South Korea}
\affiliation{Department of Materials Science and Engineering, Virginia Tech, Blacksburg, VA 24061, USA}

\date{\today}


\begin{abstract}
Quantum nonlocality is often judged by violations of Bell-type inequalities for a given state. The computation of such violations is a global task, requiring evaluation of global correlations and subsequent testing against a Bell functional. We ask instead: when is a given state local (classical)? We formalize this question via local perception operators (LPOs) that compress global observables into locally accessible statistics, and we derive two complementary witnesses---one implementable by a single party with classical side information, one intrinsically two-sided. These tools revisit familiar Bell scenarios from a new operational angle. 
We show how the witness leads to state-aware constraints that depend on local marginals and measurement geometry, with natural specializations to canonical scenarios. The resulting criteria are built from first moments and standard projective measurements and provide a way to certify compatibility with local hidden variable explanations for the LPO-processed data in regimes where conventional Bell violations may be inconclusive. 

\end{abstract}

\maketitle
\section{\label{sec:intro}Introduction}
Identifying whether a given quantum state accommodates a local hidden variable (LHV) description is one of the fundamental quests of the studies of quantum foundation. Various tests exist in this regard, the famous of the lot the Bell-CHSH test, an inequality first introduced by \citet{bell_1964_einstein}, and then reformulated to be tested on a two qubit system by Clauser-Horne-Shimony-Holt (CHSH)~\cite{clauser_1969_proposed}. Any state that satisfies the bound set by this inequality has a LHV description, and cannot be described by one if the inequality is violated. This violation was experimentally first tested by \citet{aspect_1982_experimental} albeit with some `loopholes', for the famous loophole free experiment for the same, see~\citet{hensen_2015_LoopholefreeBell}. These pathbreaking works paved the way for exploration and characterization of the quantum-ness of states. This involved proposing more tests to characterize the said quantum-ness of general multipartite state~\cite{mermin_1990_extreme}, providing more than two measurement settings to test the hidden structure of these states~\cite{collins_2004_relevant}, generalization of Bell-type inequality~\cite{collins_2002_belltype} among a sea of other works (for an overall review, see~\cite{brunner_2014_bella, cavalcanti_2016_Quantumsteering}). \citet{popescu_1995_bells} introduced the concept of super-quantum correlations, that not only maintain no-signaling, but also produce violations more than that prescribed by the traditional Bell-type inequalities. While all of these tests serve a purpose of deciding whether or not a given state is quantum by noticing violations, we can't be so certain of their nature if the violations don't happen. As it was noted by \citet{collins_2004_relevant}, and \citet{tavakoli_2020_does}, we see that there exist states, which suggest that a given state has LHV description under traditional Bell-CHSH test, can be actually non-local, and LHV model cannot reproduce the full quantum statistics. This renders the identification of states that have LHV description ambiguous, and dependent on multiple tests. 

Another approach of identifying whether a composite state is entangled (more quantum) or separable is by utilizing the Peres-Horodecki criterion, which considers the non-negativity of the partially transposed density matrix to be necessary criteria for separability~\cite{peres_1996_separability,horodecki_1997_separability}. Various other such separability based criteria exist, and they allude to the difficulty of using these criteria for higher dimensions and increasing partitions of a given composite state. In this paper, we attempt to address this gap. We introduce a ``witness'' that assesses the degree of separability of a given composite state. Using this. we can identify a wide class of states that accommodate LHV description, and can be extended to multipartite states of arbitrary dimensions. The LPO witnesses (symmetric and asymmetric) introduced here do not detect entanglement \emph{per se}. Rather, they bound what can be inferred from local marginals about the value of a chosen Bell functional after LPO processing. Accordingly, the asymmetric witness supplies a sufficient condition for LHV‑compatibility of the LPO‑processed statistics, whereas the symmetric witness furnishes state‑dependent upper bounds without constituting a Bell inequality.

The paper is structured as follows---in section~\ref{sec:lpo}, we introduce the local perception operator and discuss its key properties. In section~\ref{sec:witness}, we introduce our witness and prove the important bounds. In section~\ref{sec:numerics}, we present our numerical results, and we conclude in section~\ref{sec:conclude}.

\section{\label{sec:lpo} Local perception operator}
Consider a bipartite system \(AB\) in the Hilbert space \(\Hil_{AB} = \Hil_A\otimes\Hil_B\). For every density matrix $\rho \in \Hil_{AB}$, we can define the reduced local operator of subsystem $\J$ as $\rho_{\J} := \PTr{\Jbar}{\rho}$, where $\Jbar$ denotes the complementary subsystem of $\J$. For a given operator $X\in \mathcal{L}(\Hil_{AB})$\footnote{$\mathcal{L}(\Hil_{AB})$ is the set of all bounded linear operators acting on the Hilbert space $\Hil_{AB}$. }, a `local perception operator' (LPO) can be defined as: 
\begin{equation}\label{eq:local_perception_operators}
	(X)^\J_\rho  = \Tr_\Jbar[(I_\J{\otimes}\rho_\Jbar)X] \,\in \mathcal{L}(\Hil_\J).
\end{equation}
This LPO was introduced by Beretta in Refs.~\cite{beretta_2009_nonlinear,*beretta_2010_maximum}. Although, primarily used to define local in the context of nonlinear thermodynamic modeling of composite quantum systems, by definition, LPO presents a `classical' partitioning of a  composite system. It has the following properties.

\begin{property}   
Linearity --- for a fixed $\rho$, 
    \begin{equation}
        \label{eq:linearity}
        (\alpha X + \beta Y)^{\J}_\rho = \alpha (X)^{\J}_\rho + \beta (Y)^{\J}_\rho, \tag{P.I.}
    \end{equation}
\end{property}

\begin{property}
$\rho \mapsto (X)^{\J}_\rho$ is an affine map in the marginal $\rho_\Jbar$, but nonlinear in $\rho$.
\end{property}

\begin{property}
Classically communicated information --- $(X)^{\J}_\rho$ contains the information about $X$ that subsystem $\J$ can infer after classically receiving $\rho_{\Jbar}$,
    \begin{equation}\label{eq:local_perception_operators_id}
	\Tr[\rho_{\J}(X)^{\J}_\rho] =\Tr[(\rho_{\J}{\otimes}\rho_{\Jbar})X] =\Tr[\rho_{\Jbar}(X)^{\Jbar}_\rho]\,. \tag{P.II.}
    \end{equation}
\end{property}

\begin{property}
Uniqueness --- 
    \begin{lemma}
        For a fixed global $X$, the LPO $(X)^{\J}_\rho$ given by Eq.~\eqref{eq:local_perception_operators} is unique local operator satisfying property 3 Eq.~\eqref{eq:local_perception_operators_id} for all bi-partitions of $\rho$.
    \end{lemma}
    \begin{proof}
    Let's assume there is another operator $Y\in \mathcal{L}(\Hil_{\J})$ that satisfies property~\eqref{eq:local_perception_operators_id}, implying $\Tr(\rho_{\J}Y) = \Tr[(\rho_{\J}\otimes\rho_{\Jbar})X]$. If we define, $\Delta:= Y - (X)^{\J}_\rho$, we immediately have $\Tr(\rho_{\J}\Delta) =0$ for all bi-partitions $\rho_{\J},~ \rho_{\Jbar}$ of $\rho$. Which implies $\Delta=0$.
    \end{proof}
\end{property}

\begin{property}
No-signaling --- $(X)^{\J}_\rho$ is invariant under unitary local operations in the subsystem $\Jbar$~\cite{ray_2025_nosignalinga}. 
\end{property}
For any function $X(\rho)$ (linear or nonlinear), the quantity $(X(\rho))^A$ is independent of local unitary operations in $B$, such that $(X(\rho))^A = (X(\rho^\prime))^A$, for $\rho^\prime  = (I_A\otimes U_B)\rho(I_A\otimes U_B^\dagger)$ where $U_B$ is a random unitary acting on subsystem $B$.

\begin{property}
Preserving POVMs---
\begin{theorem}
    Let $\{ E_\alpha\}$ be a POVM over $\Hil_A\otimes\Hil_B$ with $\sum_\alpha E_\alpha = I_{AB}$. For any bipartite state $\rho$ with marginals $\rho_A$ and $\rho_B$, we have 
    \begin{equation}\label{eq:theorem_povm}
        \sum_\alpha (E_\alpha)_\rho^A = I_A, \qquad \sum_\alpha (E_\alpha)_\rho^B = I_B
    \end{equation}
    where $(E_\alpha)^A_\rho := \PTr{B}{(I_A\otimes\rho_B) E_\alpha}$ and similarly for $B$.
\end{theorem}
\begin{proof}
Using linearity of partial trace we can write:
\begin{align}
    \begin{split}
        \sum_\alpha (E_\alpha)_\rho^A &= \sum_\alpha \PTr{B}{(I_A\otimes\rho_B) E_\alpha},\\
        &= \PTr{B}{(I_A\otimes\rho_B)\sum_\alpha E_\alpha},\\
        &= \PTr{B}{(I_A\otimes\rho_B)} \;= I_A
    \end{split}
\end{align}
And similarly for $(E_\alpha)_\rho^B$.
\end{proof}
\begin{corollary}
    For any POVM $\{E_\alpha\}$ and state $\rho$, the LPO based expectations (probabilities) $p_\alpha^{(A)}:= \Tr(\rho_A (E_\alpha)_\rho^A)$ satisfy $\sum_\alpha p_\alpha^{(A)} = 1$.
\end{corollary}
\end{property}

\section{\label{sec:witness}LPO witness and its properties}
Consider $A_x$, $B_y$ to be dichotomic observables with spectra $\{\pm 1\}$ (POVM, or projective, as required from the context). We define the two-party linear Bell operator $\mathcal{B}$ as
\begin{equation}\label{eq:Bell_operator}
    \mathcal{B} = \sum_{x,y}\alpha_{xy} A_x\otimes B_y + \sum_x \beta_x A_x\otimes I + \sum_y \gamma_y I\otimes B_y.
\end{equation}
Note, this formalism allows us to express CHSH, $I_{3322}$ [or $C_{3322}$ in our case, see Eq.~\eqref{eq:I3322_corr}], and many such generalized Bell-type inequalities. We can further write 
\begin{align}\label{eq:lpo_prod_ident}
    (A_x \otimes B_y)^A_\rho & = \PTr{B}{(I\otimes \rho_B)(A_x \otimes B_y)} = A_x\underbrace{\Tr(\rho_B B_y)}_{b_y}\notag\\
    (A_x \otimes B_y)^B_\rho & = \underbrace{\Tr(\rho_A A_x)}_{a_x} B_y.
\end{align}
using Eq.~\eqref{eq:lpo_prod_ident}, we can write 
\begin{equation}\label{eq:lpo_corr_prod}
    (A_x \otimes B_y)^A_\rho \otimes (A_x\otimes B_y)^B_\rho = a_xb_y (A_x\otimes B_y).
\end{equation}
We define the following witness
\begin{equation}
    \wit{Bell} := \sup_{A_x, B_y}\Tr(\rho\mathcal{B}),
\end{equation}
where $\sup_{A_x, B_y}$ implies the witness is optimized over all measurement settings~\cite{navascues_2007_BoundingSet}.
Thereafter, the definition of witness involving LPOs follow naturally as

\begin{align}
\begin{split}
    \wit[a]{LPO} &:= \sup_{A_x,B_y}\Tr(\rho_A(\mathcal{B})^A_\rho) = \sup_{A_x, B_y}\Tr(\rho_B (\mathcal{B})^B_\rho)\\
    & = \sup_{A_x,B_y}\left[\sum_{x,y}\alpha_{xy}a_x b_y + \sum_x\beta_x a_x+ \sum_y\gamma_y b_y\right],
\end{split}\label{eq:asym_lpo_wit}
\end{align}
which simplifies to 
\begin{equation}
    \wit[a]{LPO} = \sup_{A_x,B_y}\Tr[(\rho_A\otimes\rho_B)\mathcal{B}],
\end{equation}
where superscript `a' stands for asymmetric\footnote{Note that it is asymmetric because the witness can be expressed using property 3 Eq.~\eqref{eq:local_perception_operators_id} as, $\Tr(\rho_A (\mathcal{B})^A_\rho) = \Tr(\rho_B (\mathcal{B})^B_\rho)$---any single party can compute the whole witness.}. In this case, it is equivalent to computing the Bell operator on the product of the marginals $\rho_A\otimes\rho_B$. No factorization of $\rho$ is assumed. The point is that the scalar ignores inter-party correlations in $\rho$; effectively erasing quantum contributions. It remains locally accessible since it is determined entirely by locally perceived correlators together with the local marginal, and is computable by each subsystem\footnote{We mean that it can be computed using local measurements in both $A$ and $B$ along with classical communication of their results.}.
We also note that for local observables with outcomes restricted to $\{\pm 1\}$, the expectation values also lie within the same limit. Therefore, recalling the definitions of Eq.~\eqref{eq:lpo_prod_ident}, we can write 
\begin{equation}\label{eq:bounds_det}
    a_x = \Tr(\rho_A A_x) \in [-1,+1], ~~ b_y = \Tr(\rho_B B_y)\in [-1, +1],
\end{equation}
along with
\begin{equation}\label{eq:cxy}
    C_{xy} = \Tr(\rho A_x\otimes B_y).
\end{equation}
So, Alice can collect her means in a vector $\bm{a}:= (a_x)_{x=1}^m$, and Bob can do the same in vector $\bm{b}:= (b_y)_{y=1}^n$. This allows us to write the following 
\begin{equation}\label{eq:bilinear_form}
    F(\bm{a},\bm{b}) := \bm{a}^T \bm{\alpha} \bm{b} + \bm{\beta}^T \bm{a} + \bm{\gamma}^T \bm{b}
\end{equation}
which is bilinear and affine in $\bm{a}$ and $\bm{b}$. Where $\bm{\alpha}\;|\;[\bm{\alpha}]_{xy} = \alpha_{xy}$ is a rectangular matrix of size $m\times n$, and $\bm{\beta} = (\beta_x)$ and $\bm{\gamma} = (\gamma_y)$ are column matrices of size $m\times 1$ and $n\times 1$, respectively.
We now state the following important theorem
\begin{theorem}
    Let $L_\text{LHV}(\mathcal{B})$ denote the deterministic local hidden variable bound on $\mathcal{B}$, then 
    \begin{equation*}
    \wit[a]{LPO} \leq L_\text{LHV}(\mathcal{B})
    \end{equation*}

\end{theorem}\label{th:supremum_linear}
\begin{proof}
Let's define quantum feasible sets of local means
\begin{align}
    \mathcal{K}_A &:= \{\bm{a}\in \mathbb{R}^m :\exists \rho_A \text{ s.t. } a_x = \Tr(\rho_AA_x)\} \equiv f_A(\mathcal{S}(\Hil_A)),\\
    \mathcal{K}_B &:= \{\bm{b}\in \mathbb{R}^n :\exists \rho_B \text{ s.t. } b_y = \Tr(\rho_BB_y)\} \equiv f_B(\mathcal{S}(\Hil_B)).
\end{align}
Since given a $\rho$, $\rho_A$ and $\rho_B$ is guaranteed, we define $\mathcal{S}(\Hil_A)$ as a set of $\rho_A$ on $\Hil_A$. We further observe that $\mathcal{S}(\Hil_\J)$ is non-empty, convex, and have pure-states as extreme points. For a qubit $\mathcal{S}(\Hil_\J)$ is a Bloch ball. Therefore, the following maps
\begin{align}
    f_A: \mathcal{S}(\Hil_A) \to \left[ -1,1\right]^m,&~~ f_A(\rho_A):= \left(\Tr(\rho_AA_x)\right)_{x=1}^m,\label{eq:map_fa}\\
    f_B: \mathcal{S}(\Hil_B) \to \left[ -1,1\right]^n,&~~ f_B(\rho_B):= \left(\Tr(\rho_BB_y)\right)_{y=1}^n\label{eq:map_fb}
\end{align}
help us define the feasible sets as the images. And, due to these maps Eqs.~\eqref{eq:map_fa},~\eqref{eq:map_fb} the following inclusions exist
\begin{equation}
    \mathcal{K}_A \subseteq \left[-1, 1\right]^m, \quad \mathcal{K}_B \subseteq \left[-1, 1\right]^n.
\end{equation}
Since the local behaviors  form a polytope whose extreme points are deterministic assignments, maximizing an affine functional reduces to evaluating at the vertices~\cite{brunner_2014_bella, pitowsky_1991_Correlationpolytopes, fine_1982_HiddenVariables}. Hence, 
    \begin{align}
    \begin{split}
        &\wit[a]{LPO}\\
&=
\sup_{\substack{\bm{ a}\in\mathcal{ K}_A,\\ \bm{ b}\in\mathcal{ K}_B}} F(\bm a,\bm b)
\le
\max_{(\bm{ a},\bm{ b})\in{\pm1}^{m+n}} F(\bm{ a},\bm{ b})\\
&= L_{\text{LHV}}(\mathcal B).
    \end{split}
    \end{align}
Thereafter, observing that $F(\bm{a},\bm{b})$ is affine in $\bm{a}, \bm{b}$, and the outer relaxation hypercube $[-1,1]^m \times [-1,1]^n$ has vertices $\{\pm1\}^{m+n}$ corresponding to the deterministic LHV assignments \emph{i.e.,} the maximum of a convex function on a compact interval is achieved at the end points, we can prove this theorem (a detailed proof is provided in the Appendix~\ref{supp:A}). 
\end{proof}
For example, consider the CHSH scenario ($x,y \in \{0,1\}$). In this case, we have
\begin{equation}\label{eq:chsh_coeffs}
    \bm{\alpha} = \mqty(1 & 1 \\
    1 & -1), \quad \bm{\beta} = \bm{\gamma} = 0.
\end{equation}
We have, $F(\bm{a},\bm{b}) = a_0(b_0+b_1)+a_1(b_0-b_1)$. For fixed $\bm{b}$, $a_0 = \sign{b_0 + b_1}$, $a_1= \sign{b_0 - b_1}$. Therefore,
\begin{align}
\begin{split}
    \sup_{\bm{a}\in [-1,1]^2} F(\bm{a},\bm{b}) &= \abs{b_0 + b_1} + \abs{b_0 - b_1},\\
    & = 2\max \{\abs{b_0}, \abs{b_1}\} \leq 2.
\end{split}
\end{align}
Now that we have shown that $\wit[a]{LPO}$ is upper-bounded by deterministic LHV bounds, let us consider the quadratic version of this witness. We define (using Eqs.~\eqref{eq:lpo_corr_prod},~\eqref{eq:bounds_det}, and~\eqref{eq:cxy}) 
\begin{align}\label{eq:sym_lpo_wit}
\begin{split}
    \wit[s]{LPO} &:=\sup_{A_x, B_y}\Tr[\rho (\mathcal{B})^A_\rho\otimes(\mathcal{B})^B_\rho], \\
    & = \sup_{A_x, B_y}\left[\sum_{x,y}\alpha_{xy}a_xb_y C_{xy}+\sum_x \beta_x a^2_x + \sum_y \gamma_y b^2_y\right].
\end{split}
\end{align}
where the `s' in the superscript stands for symmetric\footnote{Here we use the term symmetric because the above expression can be written as $\Tr(\rho(\mathcal{B})^A_\rho\otimes(\mathcal{B})^B_\rho)$.}. We have used the identity $\Tr(\rho (M\otimes I)) = \Tr(\rho_A M)$. Unlike $\wit[a]{LPO}$, $\wit{LPO.s}$ is nonlinear in $\bm{a},\!\bm{b}$ and for product states is quadratic in them. Considering $\rho = \rho_A\otimes\rho_B$ with Eq.~\eqref{eq:sym_lpo_wit} ---
\begin{equation}\label{eq:quadratic}
    \wit[s]{LPO} = \sup_{A_x, B_y}\left[\sum_{x,y}\alpha_{xy}a^2_xb^2_y + \sum_x \beta_x a^2_x + \sum_y \gamma_y b^2_y\right].
\end{equation}
We note that while $\wit[a]{LPO}$ can be computed locally, $\wit[s]{LPO}$ on the other hand is a global witness because it requires us to compute $\Tr[\rho(A_x\otimes B_y)]$. Because of the nonlinear dependence on state-dependent terms, compared to $\wit[a]{LPO}$, computing generic LHV bounds on $\wit[s]{LPO}$ is not trivial. However, we can provide upper and lower bounds for general states. 

Consider $A_x = \bm{a}_x\cdot\bm{\sigma}$, and $B_y = \bm{b}_y\cdot\bm{\sigma}$. Also, let us denote the local Bloch-vector of subsystem $\J$ by $\bm{r}_\J$, so that the full two-qubit system can be written as 
\begin{equation}
    \rho = \frac{1}{4}\left(I\otimes I + \bm{r}_A\cdot\bm{\sigma}\otimes I + I \otimes \bm{r}_B\cdot\sigma + \sum_{i,j}\mathcal{T}^{ij}_\rho\sigma_i\otimes\sigma_j\right).
\end{equation}
Therefore, we can write the following
\begin{align}
    \begin{split}
        a_x &= \bm{r}_A\cdot\bm{a}_x, ~~ b_y = \bm{r}_B\cdot\bm{b}_y,\\
         C_{ij} &= \Tr(\rho A_x\otimes B_y) = \mathcal{T}_{\rho}^{ij}.
    \end{split}
\end{align}

We can have the following general geometry-free bound and non-negativity property of $\wit[s]{LPO}$,
\begin{theorem}\label{th:geometry_free_bound}
    For any two qubit state $\rho$, with measurement settings $A_x$ and $B_y$ with $x=1,\cdots,m$ and $y=1,\cdots,n$, we have
    \begin{align*}
        0 \leq \wit[s]{LPO}&\leq \norm{\bm{r}_A}\norm{\bm{r}_B}\sum_{x,y}\abs{\alpha_{xy}}\\
        &+ \norm{\bm{r}_A}^2\sum_x \abs{\beta_x}+ \norm{\bm{r}_B}^2\sum_y\abs{\gamma_y}.
    \end{align*}
\end{theorem}
\begin{proof}
    This lower-bound proof is trivial, the bound exists for choices of $A_x$ and $B_y$ such that $a_x = b_y =0$. For the second part, we note that each of marginal averages $a_x$ and $b_y$ are upper bounded by corresponding $\norm{\bm{r}_A}$ and $\norm{\bm{r}_B}$ (by definition), and $\abs{C_{xy}}\leq 1$. Thereafter, the upper-bound proof follows from triangle inequality (see Appendix~\ref{supp:B} for details).
\end{proof}
We see that for CHSH, from Eq.~\eqref{eq:chsh_coeffs}, and noting $m=n=2$, we get 
\begin{equation}
    \wit[s]{LPO}\leq 4\norm{\bm{r}_A}\norm{\bm{r}_B} \leq 4.
\end{equation}
This upper-bound can be made more tight. For that we consider the following scenario, each parties' measurement settings are mutually orthogonal in corresponding real space. For this work, we will deal with at most $m=n=3$, therefore, all the $A_x$ represent mutually orthogonal directions in $\mathbb{R}^3$, and same for the $B_y$s. Under this assumption, the following theorem holds.
\begin{theorem}\label{th:supremum_ortho}
Given Alice and Bob use mutually orthogonal measurement settings, we can have\footnote{We use the subscript ``$\perp$'' to denote optimization restricted to mutually orthogonal local Bloch directions for each party.}
\begin{align*}
\begin{split}
    \wit[s]{LPO, $\perp$}&\leq \norm{\abs{\bm{\alpha}}}_2 \norm{\bm{r}_A}\norm{\bm{r}_B} + \max(\beta_x)_+\norm{\bm{r}_A}^2\\
    &+ \max(\gamma_y)_+\norm{\bm{r}_B}^2,
\end{split}
\end{align*}
where, $\abs{\bm{\alpha}}$ denotes entry-wise absolute value, $\norm{\cdot}_2$ denotes spectral norm on matrices (largest singular value), and Euclidean norm for vectors. We use the notation $(a)_+ := \max\{a,0\}$. 
\end{theorem}
\begin{proof}
     Let's define the matrix $\bm{C}$ such that $[\bm{C}]_{x,y} = C_{xy}$. Therefore, we can write from Eq.~\eqref{eq:sym_lpo_wit} 
    \begin{equation}
        \sum_{x,y}\alpha_{xy}a_xb_y c_{xy} = \bm{a}^T (\bm{\alpha}\circ \bm{C})\bm{b},
    \end{equation}
    where $\circ$ denotes Hadamard product of the two matrices. We further note that $\abs{c_{xy}}\leq 1$, thus we can write
    \begin{align}\label{eq:corr_theorem}
    \begin{split}
        \bm{a}^T (\bm{\alpha}\circ \bm{C})\bm{b} &\leq \abs{\abs{\bm{a}}^T\abs{\bm{\alpha}}\abs{\bm{b}}},\\
        & \leq \norm{\abs{\bm{a}}^T}_2\norm{\abs{\bm{\alpha}}}_2\norm{\abs{\bm{b}}}_2.
    \end{split}
    \end{align}

    Now, under the orthogonality assumption, using Pythagoras's theorem, we get
    \begin{equation}
        \norm{\bm{a}}_2^2 = \sum_x a_x^2 \leq \norm{\bm{r}_A}^2, ~~\norm{\bm{b}}_2^2 = \sum_y b_y^2 \leq \norm{\bm{r}_B}^2.
    \end{equation}
    Therefore, $\norm{\abs{\bm{a}}}_2\leq \norm{\bm{a}}_2 \leq \norm{\bm{r}_A}$ and $\norm{\abs{\bm{b}}}_2\leq \norm{\bm{b}}_2 \leq \norm{\bm{r}_B}$, which implies from Eq.~\eqref{eq:corr_theorem}
    \begin{equation}
        \bm{a}^T (\bm{\alpha}\circ \bm{C})\bm{b} \leq \norm{\abs{\bm{\alpha}}}_2\norm{\bm{r}_A}\norm{\bm{r}_B}.
    \end{equation}
    Now, because $a_x^2 \geq 0$ and $\sum_x a_x^2 \leq \norm{\bm{r}_A}^2$ and similarly for $b_x$, the maximum of $\sum_x \beta_x a_x^2$ and $\sum_y \gamma_y b_y^2$ over orthogonal setting is 
    \begin{equation}
    \bm{\beta}^T\bm{a}\leq \max(\beta_x)_+\sum_x a_x^2 \leq \max(\beta_x)_+ \norm{\bm{r}_A}^2,
    \end{equation}
    and similarly for $\bm{\gamma}^T\bm{b}$.
    Therefore, combining all of these, we obtain the required upper-bound.    
\end{proof}
Again, if we consider CHSH [Eq.~\eqref{eq:chsh_coeffs}], we get 
\begin{equation}
    \wit[s]{LPO, $\perp$}\leq 2\norm{\bm{r}_A}\norm{\bm{r}_B}\leq 2.
\end{equation}
Also, if the $\J^\Th$ marginal is maximally mixed \emph{i.e.,} $\norm{\bm{r}_\J}=0$, we find the witness value 0. Therefore, we see under mutually orthogonal setting condition, the upper bound is tighter and it is 2. In fact, we will now see, that for pure states and other various states, this bound remains the numerically achievable bound. Therefore, we conjecture that for all other conceivable measurement settings, the upper-bound of $\wit[s]{LPO}$ saturates at 2.

\section{\label{sec:numerics}Numerical results}
First, let us compute the value of LPO witness (LPOW)\footnote{Please note, this means a functional of locally perceived first‑moment statistics. It is not an entanglement witness and not a Bell inequality; it evaluates a Bell‑type functional after the LPO map has discarded inter‑party correlations.} for a pure state and maximally mixed state, which will numerically validate our theorem~\ref{th:supremum_ortho}. Firstly, we can use the Bell measurement settings for the state 
\begin{equation} 
\ket{\psi}= \dfrac{1}{\sqrt{2}}\left(\ket{01}-\ket{10}\right)
\end{equation}
with the measurement settings---
\begin{align}\label{eq:bell_measurement}
\begin{split}
    A_1 = \sigma_x,~&~ A_2 = \sigma_z,\\
B_1 = \frac{1}{\sqrt{2}}(\sigma_x+\sigma_z),~&~ B_2=\frac{1}{\sqrt{2}}(\sigma_x-\sigma_z).
\end{split}
\end{align}
We consider a standard bipartite system $AB$ with observables $A_i$ for Alice and $B_j$ for Bob. Using the usual correlation function $C_{ij}$ for LPO, we consider perception of correlation, denoting $\mathcal{X}_{ij}=A_i\otimes B_j$, we define the following:
\begin{align}\label{eq:perceived_global}
    \begin{split}
        (\mathcal{X}_{ij})_\rho^A &= A_ib_j\,,\\
        (\mathcal{X}_{ij})_\rho^B &= a_iB_j\,.
    \end{split}
\end{align}
Using this definition, we can redefine new LPO based correlation $\clpo{ij}$ as
\begin{equation}\label{eq:lpo_correlation}
    \clpo{ij} = \Tr(\rho\,( \mathcal{X}_{ij})_\rho^A\otimes (\mathcal{X}_{ij})_\rho^B) = a_i b_j C_{ij}.
\end{equation}
Using the definition of $\clpo{ij}$ with the settings in Eq.~\eqref{eq:bell_measurement} we get $\wit[s]{LPO} = 0$. For book-keeping purposes we define the following
\begin{equation}
   \schsh = \wit{Bell},~~ \slpo = \wit[s]{LPO}.
\end{equation}
To compute $\slpo$ for a pure two-qubit composite (say, $\dyad{00}$) we note, that we will use Eq.~\eqref{eq:quadratic}, and we get, (since there is pure state involved, we have $\norm{\bm{r}_\J}=1$)
\begin{equation}
    \slpo = 2.
\end{equation}

To compare with the corresponding $\schsh$, we will use
the Horodecki criteria~\cite{horodecki_1995_violating} to determine whether a state accommodate local hidden variable (LHV) description, via the correlation matrix formalism. First, we define the correlation matrix as 
\begin{equation}
    \mathcal{T}_\rho^{ij} = \Tr(\rho \sigma_i\otimes\sigma_j) = \Tr(\rho\sigma_{ij}),
\end{equation}
where, $\sigma_{ij} = \sigma_i\otimes\sigma_j$. Then, the Horodecki criteria says, that if the two largest singular values ($s_1, \, s_2$) of the matrix $\mathcal{T}_{\rho}$ satisfies the following inequality, we can say that the state $\rho$ accommodates LHV description,
\begin{equation}\label{eq:horodecki_criteria}
    \sqrt{s_1^2+s_2^2}\leq 1.
\end{equation}

Interestingly, if we consider the state, $\mathrm{I}_4/4$, a two-qubit maximally mixed state, then we get $\slpo=0$. We have already shown that for any state which is locally maximally mixed state, $\slpo=0$. From this, it follows that for states like Werner states~\cite{werner_1989_quantuma}, Bell-diagonal states~\cite{lang_2010_quantum,horodecki_1997_separability}, we will have $\slpo=0$, irrespective of their Bell-CHSH values. 

Naturally, we begin to ponder, what are the other different scenarios where LPOW differ from Bell-CHSH or similar measures, and whether they provide more insights than what we already know from Bell-CHSH.

\subsection{More than two measurement settings}
To see the benefits of LPOW, we consider the generalization of the Bell-CHSH scenario. The Bell-CHSH test is the simplest setting with two dichotomic measurement setup for Alice and Bob each. However, the Bell-CHSH test may not be sufficient to test for nonlocality in states. Some states may exist that shows no-violation under Bell-CHSH test, but other tests reveal their nonlocal character~\cite{collins_2004_relevant}. We can use the three dichotomic measurement settings for each Alice and Bob and reveal the hidden nonlocal behavior of a given state which is not revealed under a standard Bell-CHSH test~\cite{collins_2004_relevant, tavakoli_2020_does}. Following convention, we call it the $I_{3322}$ test. It can be written as:
\begin{align}\label{eq:i3322_pr}
    \begin{split}
        I_{3322} &= P(A_1\,B_1) + P(A_1\,B_2) + P(A_1\,B_3) + P(A_2\,B_1)\\
        &+ P(A_2\,B_2)  - P(A_2\,B_3)+ P(A_3\,B_1) - P(A_3\,B_2)\\
        &- P(A_1) - 2P(B_1) - P(B_2) \leq 0,
    \end{split}
\end{align}
where \( P(A_i\,B_j)\) is the probability that both the outcomes of Alice and Bob are zero. Using the correlation functions we can express the conditional joint probabilities as 
\begin{align}\label{eq:cond_joint_pr}
\begin{split}
    &P(a,b|A,B) \\ &= \dfrac{1}{4}\left(1+(-1)^aC_{AI}+(-1)^bC_{IB}+(-1)^{a+b}C_{AB} \right).
\end{split}
\end{align}
For our case, \(a=b=0\). We use \( C_{AI} = \Tr(\rho A\otimes I)\) and \( C_{IB} = \Tr(\rho I\otimes B)\). Using this, we can rewrite the expression for \(I_{3322}\) in form of correlators as~\cite{sliwa_2003_symmetries}
\begin{align}\label{eq:I3322_corr}
    \begin{split}
        C_{3322} &= C_{11} + C_{12} + C_{13} + C_{21}+ C_{22}-C_{23} \\
        &+C_{31} - C_{32} + C_{1I} + C_{2I}- C_{I1} - C_{I2} \leq 4.
    \end{split}
\end{align}
It can be shown, that the corresponding $\wit[s]{LPO, 3322}$ defined by (using Eq.~\eqref{eq:I3322_corr})
\begin{align}\label{eq:3322}
\begin{split}
\wit[s]{LPO, 3322} &= \sup_{A_x, B_y}\left(\bm{a}^T \bm{\alpha} \bm{b} + \bm{\beta}^T \bm{a} + \bm{\gamma}^T \bm{b}\right),\\
\text{with }&  \bm{\alpha} = \mqty(1 & 1 & 1\\
    1 & 1 & -1\\
    1 & -1 & 0), ~~ \bm{\beta} = \mqty(1 \\ 1 \\ 0), ~~ \bm{\gamma} = \mqty(-1 \\ -1 \\ 0),
\end{split}
\end{align}
has a tight upper-bound at $(1+\sqrt{3})\norm{\bm{r}_A}\norm{\bm{r}_B}+\norm{\bm{r}_A}^2$ (see Appendix~\ref{supp:C}).
\begin{figure}[t!]
    \centering
    \includegraphics[width=\columnwidth]{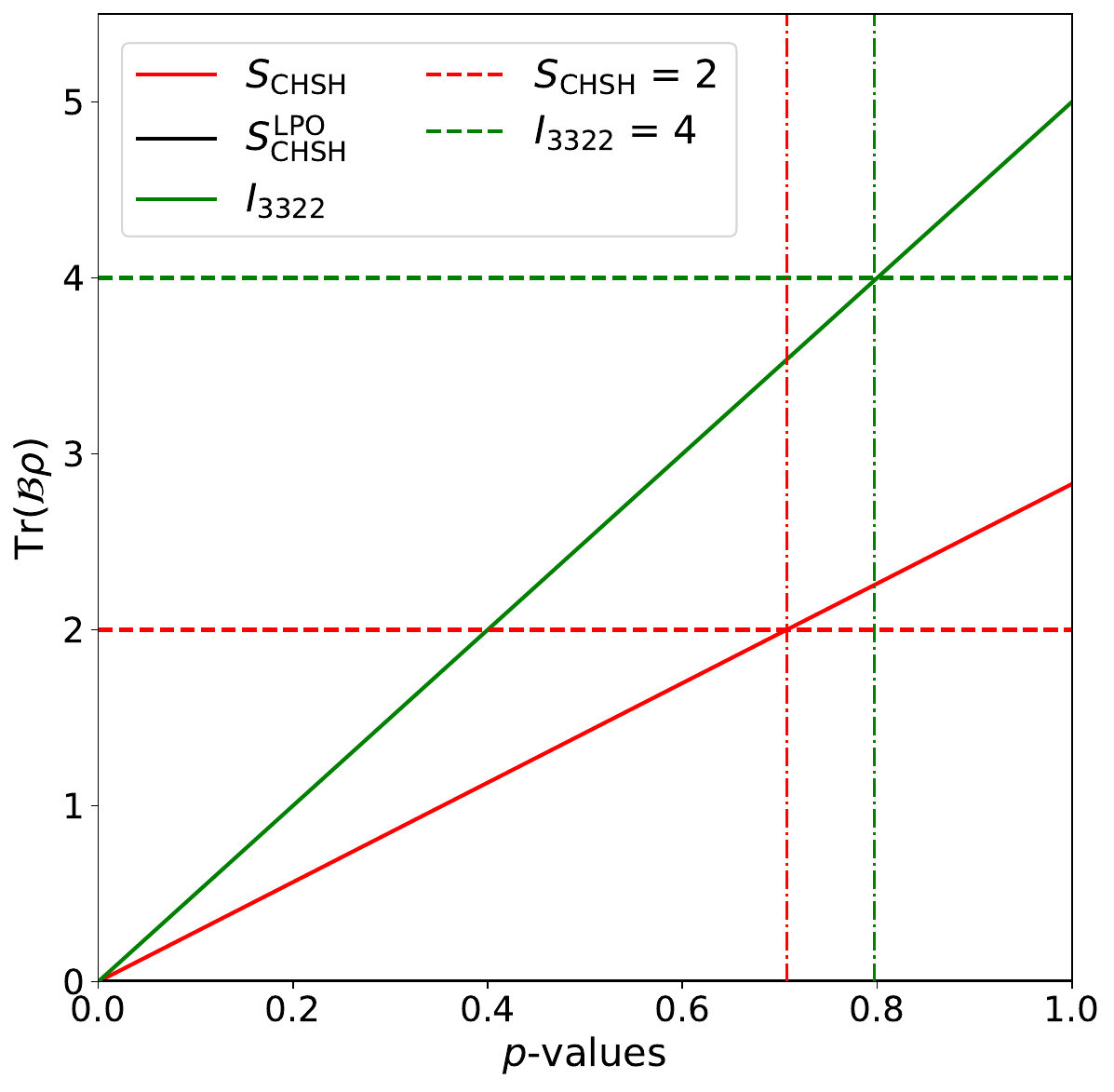}
    \caption{The witness values $\Tr(\mathcal{B}\rho)$ for various inequalities --- $I_{3322}$, $\schsh$, and $\slpo$ plotted against changing $p$ values of Werner ($\rho_{\text{w}}$) state as in Eq.~\eqref{eq:werner_state}. Since $\slpo$ is constant and zero, it coincides with the $x$-axis. The horizontal and vertical lines of each color denote the local bounds and the corresponding $p$ values for the corresponding inequalities, respectively.}
    \label{fig:werner}
\end{figure}
In literature, there exist certain states which, for any given measurement setting, does not violet Bell-CHSH tests. For example, consider the state~\cite{collins_2004_relevant}
\begin{align}\label{eq:sigma_state}
    \begin{split}
        \sigma &= 0.85\, \mathrm{P}_{\scalebox{0.5}{$\ket{\phi}$}} + 0.15\, \mathrm{P}_{\scalebox{0.5}{$\ket{01}$}}\,\\
        \text{with, } \ket{\phi} &= \dfrac{1}{\sqrt{5}}\left(2\ket{00}+\ket{11} \right),
    \end{split}
\end{align}
with $\mathrm{P}_{\scalebox{0.5}{$\ket{\phi}$}}$ is a projector onto the state $\ket{\phi}$. It can be shown, that for any measurement settings involving two measurements for Alice and Bob each, the $S_{\scalebox{0.5}{CHSH}}$ value is always less than 2. On the other hand, we find $\slpo=0$. However, for the given measurement settings~\cite{tavakoli_2020_does}, 
\begin{align}
    &O(\phi,\theta)  = \sin(\theta)\left(\cos(\phi)\sigma_x + \sin(\phi)\sigma_y\right)+\cos(\theta)\sigma_z,\\
    \begin{split}
        & A_1 = O(\eta,0),~ A_2 = O(-\eta,0),~ A_3 = O(-\frac{\pi}{2},0)\,,\\
        & B_1 = O(-\zeta,0),~ B_2 = O(\zeta,0),~ B_3 = O(\frac{\pi}{2},0)\,,
    \end{split}
\end{align}
with $\cos(\eta) = \sqrt{7/8}$ and $\cos(\zeta) = \sqrt{2/3}$, $C_{3322}$ from Eq.~(\ref{eq:I3322_corr}) is $\sim 4.05$. Noting that the reduced states are not maximally mixed (thus avoiding trivial zero of LPOW), we conclude that the $I_{3322}$ test shows non-local behavior undetected by Bell-CHSH but corroborated by LPOW. This is true for a class of cases, cases which are not pathologically designed. 

\begin{figure}[t!]
    \centering
    \includegraphics[width=\columnwidth]{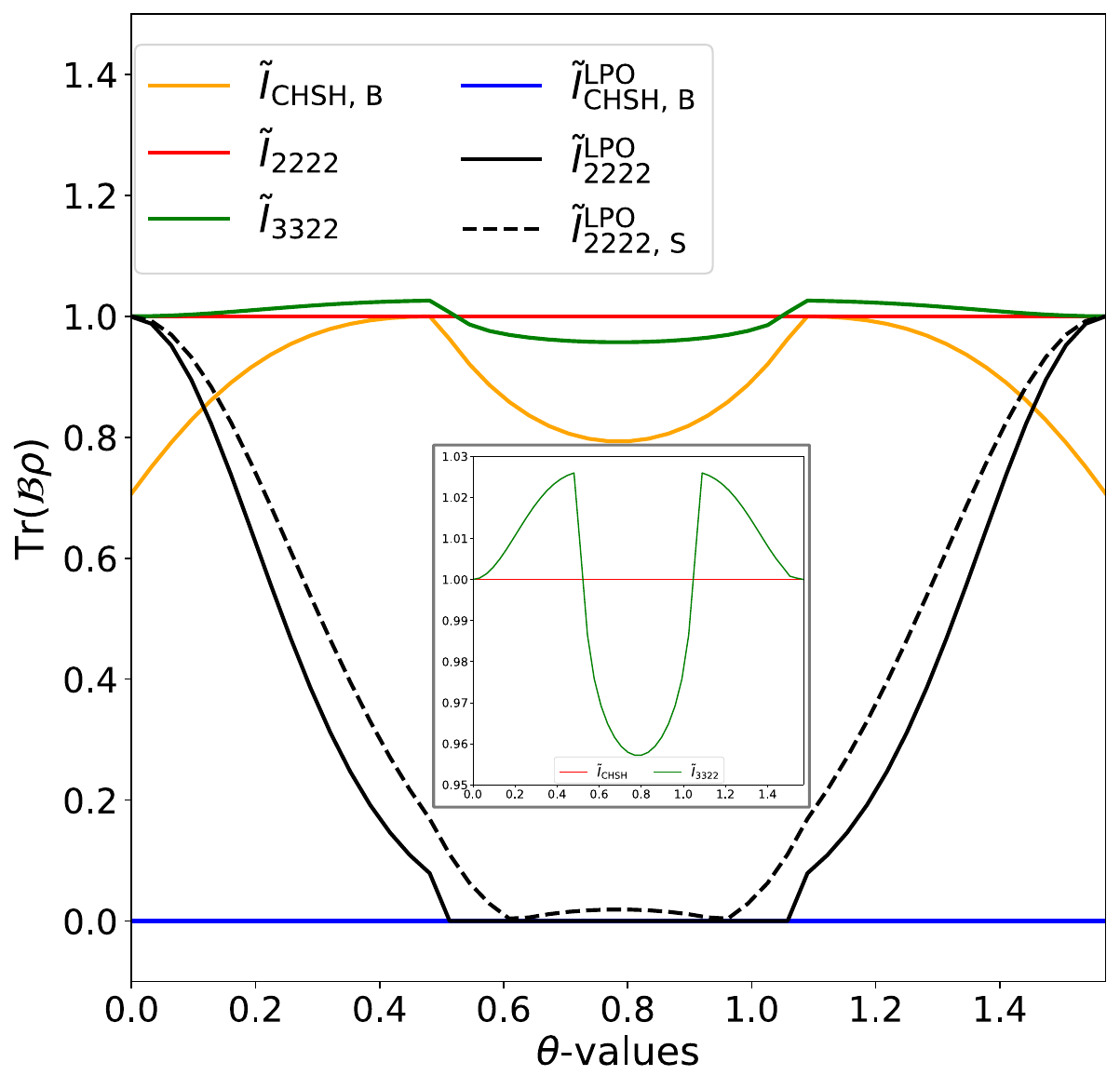}
    \caption{Value of various Bell-type inequalities for the $\ket{\text{CG}}$ (Eq.~\eqref{eq:CG_state}). The inequalities are normalized as per Eq.~\eqref{eq:chsh_coeffs}. The particular values of $\tilde{I}_{\text{CHSH, B}}$ are not optimized, and computed under standard Bell measurement setting Eq.~\eqref{eq:bell_measurement}. The inequality values labeled by $\tilde{I}^{\text{LPO}}_{\text{CHSH, B}}$ are also computed using the same settings. The black dotted curves denote inequalities computed using the settings that maximize the $\schsh$. In the inset, we have zoomed in section that corroborates the results of Ref.~\cite{collins_2004_relevant} and act as benchmark to our computation.}
    \label{fig:CG}
\end{figure}
We consider the Werner states
\begin{equation}\label{eq:werner_state}
     \rho_{\text{w}} = p\dyad{\psi}+\dfrac{1-p}{4}\mathrm{I}_4,
\end{equation}
with $\ket{\psi} = \frac{1}{\sqrt{2}}(\ket{01}-\ket{10})$ and employ the $I_{3322}$, $\schsh$, and $\slpo$ calculations, see Fig.~\ref{fig:werner}. We use $\mathcal{B}$ to denote a general inequality operator whereas the value of that inequality for a state $\rho$ is given by the witness $\Tr{\mathcal{B}\rho}$. We further compute these inequalities for a different class of states---states that have known $I_{3322}$ violation but no such violation for $\schsh$~\cite{collins_2004_relevant}. However, here, we normalize the inequalities to compare them better. Our normalization prescription is as under-
\begin{align}\label{eq:inequality_norm}
\begin{split}
    \tilde{I}_{3322} &= I_{3322} + 1, ~~\text{from Eq.~\eqref{eq:i3322_pr}},\\
    \tilde{I}_{2222} &= 2I_{2222}+1, ~~\text{from Ref.~\cite{collins_2004_relevant}},\\
    \tilde{I}^{\text{LPO}}_{2222} &= 2I^{\text{LPO}}_{2222} +1, ~~\text{from Eq.~\eqref{eq:cond_joint_pr}, and~\eqref{eq:lpo_correlation} with Ref.~\cite{collins_2004_relevant}}.
\end{split}
\end{align}
We consider the following class of states, henceforth called $\ket{\text{CG}}$ to refer to Ref.~\cite{collins_2004_relevant}, defined as 
\begin{align}\label{eq:CG_state}
    \begin{split}
         \ket{\text{CG}} &= \lambda_{\text{CHSH}}\mathrm{P}_{\ket{\theta}}+(1-\lambda_{\text{CHSH}})\mathrm{P}_{\ket{01}},\\
         \text{with, }  \ket{\theta} &= \cos(\theta)\ket{00}+\sin(\theta)\ket{11}.
    \end{split}
\end{align}
$\lambda_{\text{CHSH}}$ are chosen so that $\ket{\text{CG}}$ has maximum permissible LHV value for Bell-CHSH test (here, 1). \citet{collins_2004_relevant} showed that $\ket{\text{CG}}$ are nonlocal under $I_{3322}$ test (see inset of Fig.~\ref{fig:CG}). We additionally show that this behavior (deviation from LHV description) is also captured in the LPOW as well, see Fig.~\ref{fig:CG}.

Thereafter, we present the following two cases. In the first, we consider a two-qubit state which is known to have classicality in the Bell-CHSH and other inequalities, and we show it's behavior under the LPOW. Consider, 
\begin{equation}
\label{eq:classical}
  \ket{\theta,\beta} = \cos(\theta)\ket{00}+e^{\imi \beta}\sin(\theta)\ket{10}
\end{equation}
Here, we have parametrized the state such a way that for $\theta=0$, we have a pure state $\ket{00}$, for $\theta=\pi/2$ we have another pure state $\ket{10}$ modulated by a phase factor dependent on $\beta$. In this scenario, we see that the Bell-CHSH test and the $I_{3322}$ tests show no sensitivity to the parameters, but the LPO being inherently state dependent, shows a clear sensitivity to the parameter settings as seen in Fig.~\ref{fig:classical}.
\begin{figure}[t!]
    \centering
    \includegraphics[width=\columnwidth]{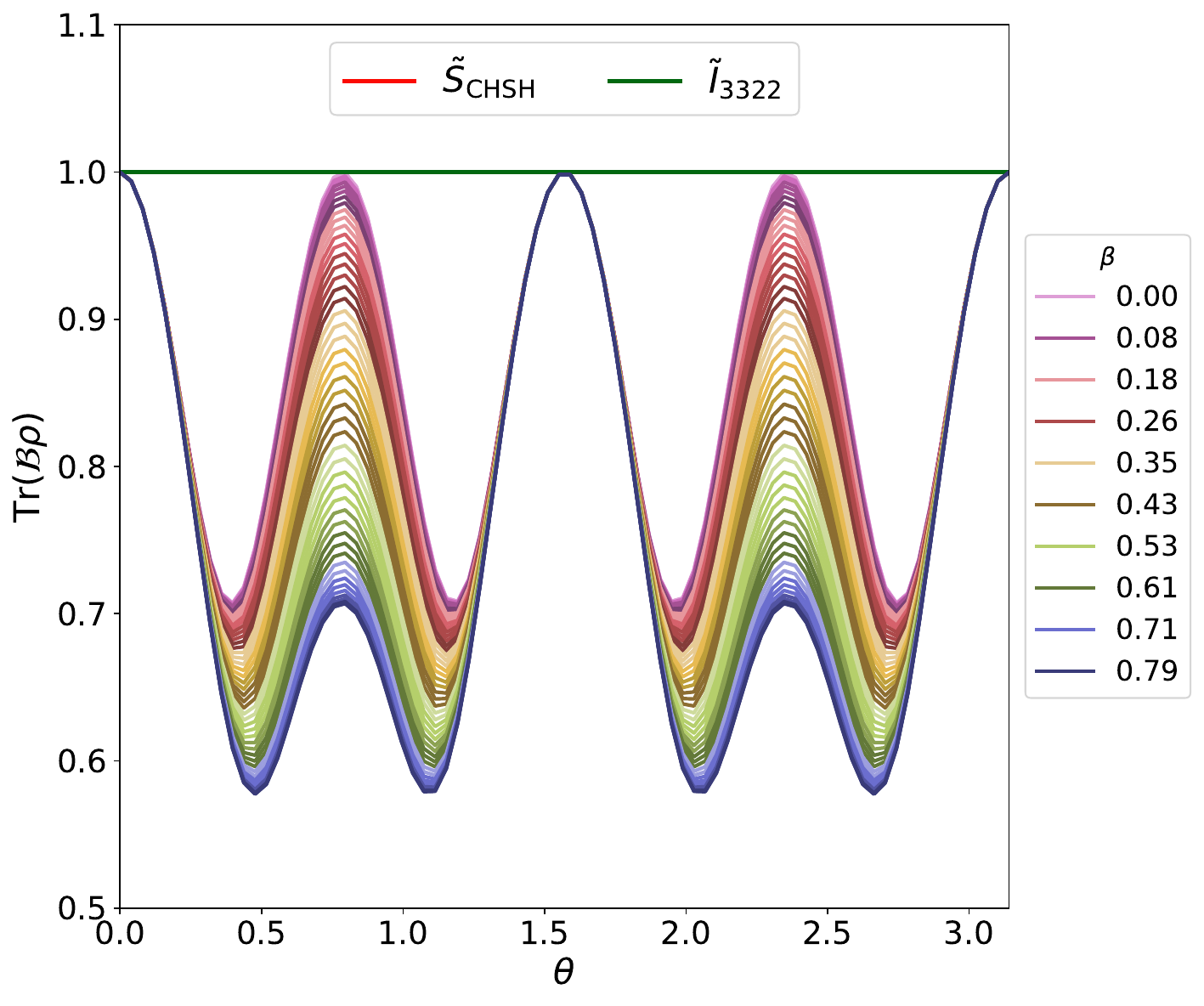}
    \caption{The normalized inequality values for the classical state given by Eq.~\eqref{eq:classical}. The plot shows that LPOW is sensitive to the internal structure of the state. Not only for pure states, but also for other superpositions it shows how the effect of phase deviates the state from pure-state behavior.}
    \label{fig:classical}
\end{figure}

In the second case, we consider a state
\begin{align}\label{eq:transition_state}
    \begin{split}
        \rho &= (1-p)\dyad{\psi^+}+ p\dyad{00},\\
        \text{with } \ket{\psi^+} &= \dfrac{1}{\sqrt{2}}(\ket{01}+\ket{10}).
    \end{split}
\end{align}
\begin{figure}[t!]
    \centering
    \includegraphics[width=\columnwidth]{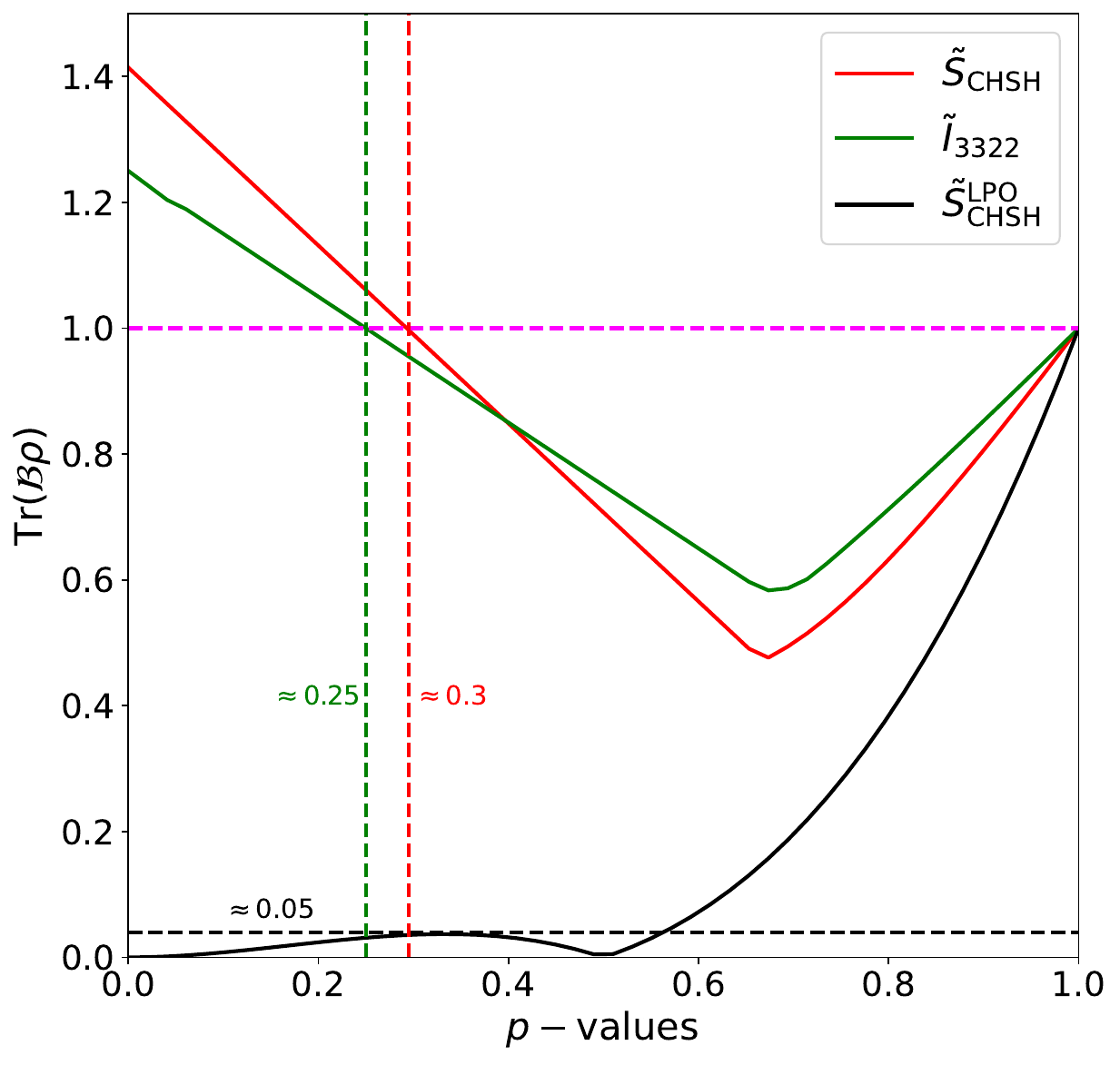}
    \caption{The normalized witness values for the states with transition as in Eq.~\eqref{eq:transition_state}. The vertical dashed lines show the transition points for CHSH and 3322 witnesses in red and green respectively (with the corresponding $p_c$ value mentioned). The solid black line represents how LPOW changes for $p$. }
    \label{fig:transition}
\end{figure}
In this case, we see for $p=0$, we have one of the triplet states, which is entangled. And for $p=1$, we have a pure state $\ket{00}$. Other than that, this particular case also suggests there is a transition point $p_c$, beyond which the state is classical. Now, if we consider the Fig.~\ref{fig:transition}, we see that according to the Bell-CHSH and $I_{3322}$ (normalized) tests, a $p-$value greater than $0.3$ and $0.25$ is required, respectively, to not violate the said inequalities. Even though LPOW starts deviating from purely quantum behavior before the $p_c$ for CHSH and $I_{3322}$ cases, it requires $p>0.5$ to convincingly avoid the state being non-classical (at $p=0.5$ the state produces locally maximally mixed states). We note that LPOW value increases slightly between $0.1\leq p < 0.5$ to a value $\approx 0.05$ before returning to $0$ again, which suggests that for these values of $p$, the state of Eq.~\eqref{eq:transition_state} may not be completely quantum. This also implies, there maybe other inequalities that may show lack of LHV theory for $p<0.5$. 
\subsection{Tripartite states}
It is not only the two-qubit systems that can be tested using LPOW. One can extend this formalism for three-qubit system as well. Here, we define our correlation operators for Alice, Bob, and Charlie as 
\begin{align}\label{eq:tr_lpo_corr}
    \begin{split}
         (X)^A_\rho &= \Tr_{BC}((I\otimes\rho_B\otimes\rho_C)X),\\
         (X)^B_\rho &= \Tr_{AC}((\rho_A\otimes I\otimes\rho_C)X),\\
         (X)^C_\rho &= \Tr_{AB}((\rho_A\otimes\rho_B\otimes I)X).
    \end{split}
\end{align}
We can define correlation operators as
\begin{align}\label{eq:mermin_measurement}
    \begin{split}
         \mathcal{X}_{xxx} = \sigma_x\otimes\sigma_x\otimes\sigma_x, ~~&~~ \mathcal{X}_{xyy} = \sigma_x\otimes\sigma_y\otimes\sigma_y,\\
\mathcal{X}_{yxy} = \sigma_y\otimes\sigma_x\otimes\sigma_y, ~~&~~ \mathcal{X}_{yyx} = \sigma_y\otimes\sigma_y\otimes\sigma_x .
    \end{split}
\end{align}
Using Mermin's inequality~\cite{mermin_1990_extreme} as under ---
\begin{align}
\begin{split}
     ~~~~\Tr(\rho\mathcal{B}_{\text{Mermin}})&\leq 2,\\
     \text{with, } \mathcal{B_{\text{Mermin}}} &= \mathcal{X}_{xxx}-\mathcal{X}_{xyy}-\mathcal{X}_{yxy}-\mathcal{X}_{yyx}.
\end{split}
\end{align}
It can be shown, that the GHZ state, $\frac{1}{\sqrt{2}}(\ket{000}+\ket{111})$ using the measurement settings of Eq.~\eqref{eq:mermin_measurement} violates Mermin's inequality with a violation of 4, however, the corresponding LPOW defined using Eq.~\eqref{eq:tr_lpo_corr} produces a value of $0$. On the other hand, for a given pure state (say, $\ket{000}$), the LPOW value is 2 (which is the classical limit of Mermin's inequality). So, LPO is a scalable method that can be incorporated to identify a given $n-$qubit state whether it accommodates a classical description or accommodates non-local theory.

\section{\label{sec:conclude}Conclusion and discussion}
To conclude, we have introduced the local perception operator (LPO) formalism to compute the locally perceived global correlators that can then be used to compute relevant Bell-type LPO witnesses (LPOW). By construction LPOs are state dependent, they project a global observable onto what is inferable from the local marginals. We introduce a one-sided (asymmetric) and two-sided (symmetric) LPOW, $\wit[a]{LPO}$ and $\wit[s]{LPO}$, respectively. The asymmetric witness is locally implementable (either by Alice or Bob---they will arrive at the same conclusion), is bilinear in the local means and affine in the coefficients, and---when optimized over settings---is constrained by the LHV bound of the underlying Bell functional, thus providing a sufficient certificate of compatibility with an LHV description for the LPO-processed data.

The symmetric witness is global and nonlinear; instead of a single state-independent benchmark, its behavior is controlled by the local marginals and the measurement geometry. We derived state-dependent, geometry-free upper bounds and sharpened them under orthogonal settings; these bounds collapse to zero when a marginal is maximally mixed. Our numerical results further suggest a CHSH-type upper-bound of 2 for the optimized symmetric case, which we state as a conjecture.

Operationally, the two witnesses complement the standard ``violation'' logic of Bell tests. When conventional violations are absent or inconclusive, $\wit{LPO}$s organize what can be concluded from locally perceived statistics alone: maximally mixed marginals lead both witnesses to vanish (hence no verdict), while nontrivial marginals are constrained by the bounds we obtain and can rule in compatibility with LHV models for the LPO-processed observables. We show that our LPOW can be extended to multipartite cases, however, proving the bounds analytically can be quite troublesome. 

Looking ahead, a natural direction is to map the boundary---within the LPO framework---between classes of entangled states (beyond maximally entangled) and states compatible with LHV models for the LPO-processed data. This includes investigating whether suitably normalized LPOWs admit a meaningful separability distance and how such indicators interact with known entanglement classes. We can also study how LPO can be implemented in task entanglement detection when there is inaccuracy in local measurements~\cite{morelli_2022_EntanglementDetection}. A complementary line is to examine the role of contextuality~\cite{budroni_2022_kochenspecker} in LPO-based witnesses, clarifying when LPO processing preserves or washes out contextual signatures.

\section*{Acknowledgments}
This work was supported by the Institute for Basic Science in Korea (IBS-R024-D1). I am grateful to Dominik \v{S}afr\'{a}nek, Gian Paolo Beretta, and Sumit Mukherjee for helpful discussions. 

\appendix

\section{\label{supp:A}Proof of theorem 2.}
We recall because $A_x, B_y \to \{\pm 1\}$, which implies from Eq.~\eqref{eq:bounds_det} of the main text, 
\begin{equation}
    \abs{\Tr(\rho_A A_x)} = \abs{a_x}\leq 1, ~~ \abs{\Tr(\rho_B B_y)} = \abs{b_y}\leq 1, ~~\forall x,y.
\end{equation}
And we further recall the definition of $F(\bm{a},\bm{b})$ [Eq.~\eqref{eq:bilinear_form}] which is bilinear and affine in $\bm{a}$ and $\bm{b}$
\begin{equation}
    F(\bm{a},\bm{b}) := \bm{a}^T \bm{\alpha} \bm{b} + \bm{\beta}^T \bm{a} + \bm{\gamma}^T \bm{b}.
\end{equation}
We note that $F(\cdot,\cdot)$ is continuous, and the quantum feasible sets  $\mathcal{K}_A \times \mathcal{K}_B$ is compact. Therefore, the supremum, 
$\sup_{\substack{\bm{a}\in \mathcal{K_A},\\ \bm{b}\in \mathcal{K}_B}}F(\bm{a},\bm{b})$
is attained. To show optimization of $F$ over $\mathcal{K}_A\times\mathcal{K}_B$ is upper-bounded by its optimization over the hypercube $\mathbb{H}:= [-1,1]^m\times[-1,1]^n$, we note that over $\mathbb{H}$, the maximum is achieved at a vertex where all coordinates equal to $\pm1$, thus, yielding LHV bound~\cite{Boyd_Vandenberghe_2004}. Hereby, we get
\begin{equation}
    \sup_{\substack{\bm{a}\in \mathcal{K_A},\\ \bm{b}\in \mathcal{K}_B}}F(\bm{a},\bm{b}) \leq \sup_{\substack{\bm{a}\in [-1,1]^m,\\ \bm{b}\in [-1,1]^n}}F(\bm{a},\bm{b}).
\end{equation}
To check whether coordinate-wise linearity and endpoint is attained, we proceed as follows. For a fixed $\bm{b}\in [-1,1]^n$, we can write
\begin{align}\label{eqS:affine_proof}
    \begin{split}
        F(\bm{a},\bm{b}) &= \sum_{x=1}^m \left(\sum_{y=1}^n \alpha_{xy}b_y + \beta_x \right)a_x + \bm{\gamma}^T\bm{b},\\
        & = \sum_{x=1}^m c_x a_x + \text{fixed}.
    \end{split}
\end{align}
This effectively shows $F(\bm{a},\bm{b})$ is affine and linear in $a_x$, and similarly for $b_y$. Because $F(\bm{a},\bm{b})$ is linear function on a closed interval $[-1,1]$, it's maximum must lie at the endpoints~\cite{brunner_2014_bella, pitowsky_1991_Correlationpolytopes, fine_1982_HiddenVariables}. For $\bm{a}$ maximization, we pick a maximizer $a^\ast_x(\bm{b})$ so that
\begin{equation}
    a^\ast_x(\bm{b}) \in \underset{a_x \in [-1,1]}{\text{argmax }} c_x(\bm{b})a_x = \begin{cases}
        +1, ~ c_x(\bm{b}) > 0,\\
        -1, ~ c_x(\bm{b}) < 0,\\
        [-1,1], ~ c_x(\bm{b}) =0.
    \end{cases}
\end{equation}
This implies, we pick $a^\ast_x(\bm{b})$ such that the argument in the summand of Eq.~\eqref{eqS:affine_proof} is always positive, and that the maximum of $a_x(\bm{b})$ will be determined by the sign of $c_x(\bm{b})$. Hence,
\begin{equation}
    \sup_{\bm{a} \in [-1,1]} F(\bm{a},\bm{b}) = \sum_{x=1}^m \abs{c_x(\bm{b})} + \bm{\gamma}^T\bm{b}.
\end{equation}
Now, we will maximize for $\bm{b}$. For each coordinate $b_y$, keeping other $b_{y'}$ ($y \neq y'$) frozen, we define $d_{xy} := \beta_x + \sum_{y \neq y'}\alpha_{xy'} b_{y'}$. Therefore, 
\begin{equation}
    c_x(\bm{b}) = d_{xy} + \alpha_{xy}b_y,
\end{equation}
which can be used to write 
\begin{equation}
    g(\bm{b}) = \left[\sum_{x=1}^m \abs{d_{xy} + \alpha_{xy}b_y}+\gamma_yb_y\right]+ \sum_{y\neq y'}\gamma_{y'}b_{y'}.
\end{equation}
Here too, maximum is achieved at the end-point $b_y = \pm 1$. We note that for each $y$, with $b_{y'\neq y}$ held constant,  $\phi_y(b_y) = \sum_{x=1}^m \abs{d_{xy} + \alpha_{xy}b_y}+\gamma_yb_y$ is convex and piecewise linear on $[-1,1]$, hence maximum is at $b_y = \pm1$. We know, a sum of convex functions with a linear term is also convex, and that a maximum of a convex function over a compact interval is achieved at the end points. We repeat the same argument for all $y$, and get a global maximizer $(\bm{a}^\ast, \bm{b}^\ast)$ of $F$ on $[-1,1]^m\times [-1,1]^n$ with all coordinates $a_x^\ast, b_y^\ast \in \{\pm1\}$, \emph{i.e.,} vertex of $\mathbb{H}$.

A vertex $(\bm{a},\bm{b})\in \{\pm1\}^{m+n}$ assigns a definite outcome $\pm1$ to each local settings on each side. Such assignments are precisely the deterministic LHV strategies~\cite{brunner_2014_bella}. 
\begin{equation}
    \max_{(\bm{a},\bm{b})\in \{\pm1\}^{m+n}} F(\bm{a},\bm{b}) = L_\text{LHV}(\mathcal{B}).
\end{equation}
Combining these we get, 
    \begin{equation}
    \wit[a]{LPO}=
\sup_{\substack{\bm{ a}\in\mathcal{ K}_A,\\ \bm{ b}\in\mathcal{ K}_B}} F(\bm a,\bm b)
\le
\max_{(\bm{ a},\bm{ b})\in{\pm1}^{m+n}} F(\bm{ a},\bm{ b})= L_{\text{LHV}}(\mathcal B).
    \end{equation}
We note, equality holds if and only if the maximizing deterministic assignment $(\bm{a},\bm{b})\in \{\pm1\}^{m+n}$ are jointly realizable as vectors of means by some local states $(\rho_A, \rho_B)$ for the given operator families $\{A_x\}, \{B_y\}$. A sufficient (but not necessary) condition is that each family $\{A_x\}$ ($\{B_y\}$) are commuting projectors, and there exists a local pure state $\ket{\phi_A}$ ($\ket{\phi_B}$) that is a common eigenvector with eigenvalues prescribed by $\bm{a}$ ($\bm{b}$). Then $\rho_A = \dyad{\phi_A}$, $\rho_B = \dyad{\phi_B}$ realize $\wit[a]{LPO} = L_\text{LHV}(\mathcal{B})$.

In general, if the local operator families are incompatible, (\emph{e.g.,} non-commuting Pauli operators on a qubit) then no single local state can realize $\abs{b_y} = 1$ for all $y$ simultaneously. In that case the LHV vertex is outside $\mathcal{K}_A\times \mathcal{K}_B$, and the inequality is strict $\wit[a]{LPO} < L_\text{LHV}(\mathcal{B})$.

\section{\label{supp:B}Proof of Theorem 3.}
Let's sort the lower bound first. As mentioned in the main text, invoking free will in the choice of $A_x$ and $B_y$, we can choose all of $A_x$ and $B_y$ such that $a_x = 0$ and $b_y = 0$ (for geometry-free optimization, with no constraints, we can do this). And thereafter, we see because of supremum in the definition, $\wit[s]{LPO} \geq 0$. This part just states the obvious---if there exists one admissible choice with value 0, then the maximum over all choices cannot be below 0 without saying anything about which settings maximize the witness. 

The upper bound is achieved as follows. For arbitrary settings, we can write (using triangle inequality)
\begin{align}\label{eqS:limits_triangle}
    \begin{split}
        \wit[s]{LPO} &\leq \sup_{A_x, B_y}\left[ \sum_{x,y}\abs{\alpha_{xy}}\abs{a_x}\abs{b_y}\abs{C_{xy}} + \sum_x\abs{\beta_x}a_x^2\right.\\
        &\left.+ \sum_y\abs{\gamma_y}b_y^2 \right],\\
        &\leq \sup_{A_x, B_y}\left[ \sum_{x,y}\abs{\alpha_{xy}}\abs{a_x}\abs{b_y} + \sum_x\abs{\beta_x}a_x^2\right.\\
        &\left.+ \sum_y\abs{\gamma_y}b_y^2 \right], ~\text{using } \abs{C_{xy}} = \abs{\Tr(\rho X_{x,y})}\leq 1.
    \end{split}
\end{align}
We now note that, from the properties of the Bloch vectors,
\begin{equation}
    \abs{a_x} \leq \norm{\bm{r}_A},~~ \abs{b_y} \leq \norm{\bm{r}_B},~~ a^2_x \leq \norm{\bm{r}_A}^2 ~~ b_y^2 \leq \norm{\bm{r}_B}^2.
\end{equation}
Substituting this directly in Eq.~\eqref{eqS:limits_triangle}, we get 
\begin{align}
\begin{split}
    \wit[s]{LPO}&\leq \norm{\bm{r}_A}\norm{\bm{r}_B}\sum_{x,y}\abs{\alpha_{xy}} + \norm{\bm{r}_A}^2\sum_x\abs{\beta_x}\\
    &+ \norm{\bm{r}_B}^2\sum_y\abs{\gamma_y}.
\end{split}
\end{align}
Thereby, proved. 

\section{\label{supp:C}$C_{3322}$ related bound}
From Eq.~\eqref{eq:bilinear_form} of main text, and the definition of $C_{3322}$ as given in Eq.~\eqref{eq:3322}, we can write the following matrices~\cite{collins_2004_relevant}
\begin{equation}
    \bm{\alpha} = \mqty(1 & 1 & 1\\
    1 & 1 & -1\\
    1 & -1 & 0), ~~ \bm{\beta} = \mqty(1 \\ 1 \\ 0), ~~ \bm{\gamma} = \mqty(-1 \\ -1 \\ 0).
\end{equation}
\begin{proposition}
    For any two-qubit state, we have for $C_{3322}$ case
    \begin{equation*}
        0 \leq \wit[s]{LPO,3322} \leq 8\norm{\bm{r}_A}\norm{\bm{r}_B} + 2\norm{\bm{r}_A}^2
    \end{equation*}
\end{proposition}
\begin{proof}
    We use the notation $(a)_+ := \max\{a,0\}$. Now we consider triangle inequality as before and start from the result of Eq.~\eqref{eqS:limits_triangle}. We find
    \begin{equation}
        \sum_{x,y}\abs{\alpha_{xy}} = 8, ~\sum_x (\beta_x)_+ = 2, ~\sum_y (\gamma)_+ = 0
    \end{equation}
    because negatives in $\beta_x$ and $\gamma_y$ never help a supremum.
    And therefore, we have the upper-bound. For lower bound, as before, we select $\bm{a}_x \perp \bm{r}_A$ and $\bm{b}_y\perp \bm{r}_B$ (see proof of Theorem~\ref{th:supremum_ortho} in the main text).
\end{proof}
    Next, we proceed to make a tighter bound. 
    \begin{proposition}
        For three mutually orthogonal Bloch directions, we will have
        \begin{equation*}
            \wit[s]{LPO,3322, $\perp$}\leq (1+\sqrt{3})\norm{\bm{r}_A}\norm{\bm{r}_B}+ \norm{\bm{r}_A}^2.
        \end{equation*}
    \end{proposition}
\begin{proof}    
    The important part of the proof lies in the following. For the correlator part, direct singular value calculation of $\bm{\alpha}$ leads to
    \begin{equation}
        \norm{\abs{\bm{\alpha}}}_2 = 1 + \sqrt{3}
    \end{equation}
    which reduces the correlator contribution to $(1+ \sqrt{3})\norm{\bm{r}_A}\norm{\bm{r}_B}$. And the rest follows trivially.
\end{proof}

\bibliography{refer}

\end{document}